\documentclass[12pt, reqno]{amsart}
\usepackage{amsmath, amsthm, amscd, amsfonts, latexsym, amssymb, graphicx, color}
\usepackage[bookmarksnumbered, colorlinks, plainpages]{hyperref}
\usepackage{tikz-cd}
\DeclareMathOperator{\ef}{\mathbb{F}}
\DeclareMathOperator{\enn}{\mathbb{N}}
\DeclareMathOperator{\zed}{\mathbb{Z}}
\DeclareMathOperator{\lc}{\leqslant_{\rm lex}}

\makeatletter \oddsidemargin.9375in \evensidemargin \oddsidemargin
\newtheorem{thm}{Theorem}[section]

\newtheorem{prop}[thm]{Proposition}

\theoremstyle{definition}

\theoremstyle{remark}

\numberwithin{equation}{section}

\begin{document}

\setcounter{page}{1}

\title[Finding a Generator Matrix]{Finding a Generator Matrix of a Multidimensional Cyclic Code}
\author[R. Andriamifidisoa, R. M. Lalasoa and T. J. Rabeherimanana]{Ramamonjy Andriamifidisoa$^{*}$, Rufine Marius Lalasoa and Toussaint Joseph Rabeherimanana}
\thanks{{\scriptsize
\hskip -0.4 true cm MSC(2010): primary: 13F20, 16D25; Secondary: 94B60
$*$ Corresponding author }}
\begin{abstract}
We generalize Sepasdar's   method for finding a gene-
\\rator matrix  of two-dimensional cyclic codes  to   find an independent subset of a general multicyclic code, which may form a basis of the code as a vector subspace. A generator matrix can be then constructed from this basis. \end{abstract}

\maketitle
\section{Introduction}
    Sepasdar, in \cite{z0} presented  a method to find a generator matrix of two dimensional skew cyclic Codes.  Then, Sepasdar and Khashyarmanesh, in \cite{z2} gave a method to find a generator matrix of some class of two-dimensional cyclic codes. Finally,  Sepasdar, in \cite{z1}, found a method to construct a generator matrix for general two-dimensional cyclic codes. In this paper,  we will generalize  this  Sepasdar's  method for a general multicyclic code.\ Our method uses an ideal basis of the code whose construction  was presented by Lalasoa et al.  in \cite{lar}.\\

In section \ref{pre} of this paper, we recall the notations used in \cite{lar} and the mathematical tools we will need, including two orderings : the partial ordering ``$\leqslant_+$" and the well ordering ``$\lc$". This latter allows to define degrees of polynomials in the quotient-ring with a special property, given by Proposition \ref{deg-prod}.\\

In section \ref{res}, we present our results. Proposition \ref{G-set} gives an idea of how a basis of the  multicyclic code, considered as a vector space will look like.\ The main result is Theorem \ref{gen-s}, which allows the construction of an independent subset of the code. Once a basis is found, one can then construct a generator matrix by forming the matrix whose rows are  the coefficients of the polynomials of the basis.

\section{Notations and Preliminaries}\label{pre}
We briefly recall the notations which are used in  \cite{lar}.   We denote the quotient ring  $\ef_q[X_1,\ldots, X_s]/\langle X_1^{\rho_1}-1,\ldots,X_s^{\rho_s}-1\rangle$ by $\ef_q[x_1,\ldots,x_s]$,  where  $\ef_q$ is the finite field with $q$ element, $x_i$ the residue class of $X_i$ modulo the ideal $\langle X_1^{\rho_1}-1,\ldots,X_s^{\rho_s}-1\rangle$.
We have
\begin{align}\label{x-puiss-rho}
 x_i^{\rho_i}=1,
\end{align}
so that
\begin{equation}\label{x-puiss}
  x_i^{m}=x_i^{m \mod \rho_i}\quad\text{for}\quad m\in\enn\quad \text{and}\quad  i=1,\ldots,s,
\end{equation}
where $m\mod \rho_i$ is the remainder of $m$ by the euclidean division of $m$ by $\rho_i$.

The additive product group $\mathcal{G}_s$ is defined by
\begin{equation*}
\mathcal{G}_s = \zed/\rho_1\zed\times\ldots\times\zed/\rho_s \zed,\\
\end{equation*}
with
\begin{equation*}
 \zed/\rho_i\zed=\{0,1,\ldots,\rho_i-1\}.
\end{equation*}
An element of $\ef_q[x_1,\ldots,x_s]$ is of the form
\begin{equation}\label{f}
  f(x_1,\ldots,x_s) =\sum_{(\alpha_1,\ldots,\alpha_s)\in\mathcal{G}_s}f_{(\alpha_1,\ldots,\alpha_s)}x_1^{\alpha_{1}}\cdots x_s^{\alpha_{s}}.
\end{equation}
For sake of simplicity, we denote $(\alpha_1,\ldots,\alpha_s)\in\mathcal{G}_s$ or, more generally, $(\alpha_1,\ldots,\alpha_s)\in\enn^s$  by $\alpha$. Then  \eqref{f} can then be written as
a
\begin{equation}\label{pol}
  f(x)=\sum_{\alpha\in\mathcal{G}_s}f_\alpha x^\alpha,
\end{equation}
where
\begin{equation}\label{x-puiss-alpha}
  x^\alpha = x_1^{\alpha_{1}}\cdots x_s^{\alpha_{s}},
\end{equation}
and we may omit the set $\mathcal{G}_s$.
For $\alpha\in\enn^s$, we also adopt the notation
\begin{equation*}  \alpha \mod \rho =(\alpha_1 \mod\rho_1,\ldots,\alpha_s\mod\rho_s)\in\mathcal{G}_s,
\end{equation*}
where $\rho=(\rho_1,\ldots,\rho_s)$. Equations  \eqref{x-puiss-rho} and \eqref{x-puiss} are then  ``generalized" to the following:
\begin{equation}\label{mult-x-rho}
  x^{\rho}=1\quad\text{and}\quad  x^\alpha= x^{\alpha\mod \rho}.
\end{equation}

The set $\enn^s$, and therefore also he product   group $\mathcal{G}_s$ is provided with two orders : a partial ordering $\leqslant_+$ defined by
\begin{equation*}
\alpha\leqslant_+\beta\iff\alpha_i\leqslant\beta_i\quad \text{for}\quad i=1,\ldots,s,
\end{equation*}
and a  \textit{well ordering} $\lc$ (the ``lexicographical ordering'), defined by
\begin{equation*}
\alpha <_{\rm lex}\beta\Longleftrightarrow\text{for the first index $i$ such that\;$\alpha_i\neq\beta_i$, one has \;$\alpha_i<\beta_i$}.
\end{equation*}
Put $n=\rho_1\cdots\rho_s$. We then may write  $\mathcal{G}_s=\{\alpha^{(1)},\ldots,\alpha^{(i)},\ldots,\alpha^{(n)}\}$ with
\begin{equation}\label{G-ordered}
\alpha^{(1)}<\cdots\alpha^{(i)}<\cdots<\alpha^{(n)}
\end{equation}
and the polynomial $f(x)$ in \eqref{pol} can  be written as
\begin{equation}\label{f-ordered}
  f(x)= f_{\alpha^{(1)}}x^{\alpha^{(1)}}+\cdots+ f_{\alpha^{(i)}}x^{\alpha^{(1)}}+\cdots+f_{\alpha^{(\rho_s)}}x^{\alpha^{(n)}}.
\end{equation}
If $f(x)$ is non-zero, we may define its \textit{degree,} denoted $\deg f(x)$ or simply $\deg f$ as
\begin{equation}\label{deg-pol}
 \deg f = \max_{\lc}\{\alpha^{(i)}\;\vert\;f_{\alpha^{(i)}}\neq 0\}.
\end{equation}
(Note that it is the usual definition of the degree of  a multivariate polynomial). However, due to equations \eqref{mult-x-rho}, for two polynomials $f$ and $g$ of $\ef_q[x_1,\ldots,x_s]$, the equality $\deg(fg)=\deg f+\deg g$ does not necessarily hold. The following proposition gives a sufficient condition for this property.
\begin{prop}\label{deg-prod}If $f$ and $g$ are non-zero elements of $\ef_q[x_1,\ldots,x_s]$ such that $\deg f+\deg g<_{+}\rho$, then $\deg(fg)=\deg f+\deg g$.
\end{prop}
\begin{proof} Write $f(x_1,\ldots,x_s)=\sum_{\alpha} f_\alpha x^\alpha$ and $ g(x_1,\ldots,x_s)=\sum_{\beta}g_\beta x^\beta$. Then, using the second equation of \eqref{mult-x-rho}, we have
\begin{align*}
f(x_1,\ldots,x_s)g(x_1,\ldots,x_s)&= \sum_\alpha\sum_\beta\ f_\alpha g_\beta x^{(\alpha+\beta)\mod\rho}\\
&=\sum_\alpha\sum_\beta\ f_\alpha g_\beta x^{(\alpha+\beta)}
\end{align*}
since $\alpha+\beta
\leqslant_+\deg f+\deg g<_+\rho=(\rho_1,\ldots,\rho_s)$ for all $\alpha$ and $\beta$. Thus
\begin{equation*}
  \deg(fg)=\max_{\lc}(\alpha+\beta)=\deg f+\deg g.
\end{equation*}
\end{proof}
All the previous results are also true for  the quotient ring
\begin{equation*}
S = \ef_q[X_1,\ldots, X_s]/\langle X_1^{\rho_1}-1,\ldots,X_{s-1}^{\rho_{s-1}}-1\rangle =\ef_q[x_1,\ldots,x_{s-1}],
\end{equation*}
with $s-1$ variables, where $x_i$ is the residue  class of  $x_i$   modulo the ideal $\langle X_1^{\rho_1}-1,\ldots,X_{s-1}^{\rho_{s-1}}-1\rangle $. Note that we have used the same notation $x_i$, because the residue class of $x_i$ modulo the ideal $\langle X_1^{\rho_1}-1,\ldots,X_{s-1}^{\rho_{s-1}}-1\rangle$ may be identified with its class modulo the ideal $\langle X_1^{\rho_1}-1,\ldots,X_s^{\rho_s}-1\rangle$, (cf. Proposition 2.2, \cite{lar}).\\

A\textit{ multicyclic code} is an ideal of $R$.\\

Let $I$ be a non-zero ideal of $R$ and
\begin{equation}\label{Ideal-base}
\mathfrak{B}=\{\mathfrak{p}^{(0)}_1,\ldots,\mathfrak{p}^{(0)}_{r_1},\mathfrak{p}^{(1)}_1,\ldots,\mathfrak{p}^{(1)}_{r_1},\ldots,\mathfrak{p}^{(i)}_1,\ldots,\mathfrak{p}^{(i)}_{r_i},\ldots,\mathfrak{p}^{(\rho_s-1)}_1,\ldots,\mathfrak{p}^{(\rho_s-1)}_{r_{\rho_s-1}}\}
\end{equation}
the basis of   $I$, found by Lalasoa et al. by the method in \cite{lar}. Then an element $f(x_1,\ldots,x_s)\in R$ may be written as
\begin{equation}\label{I-expr}
f(x_1,\ldots,x_s) = \sum_{i=0}^{r_{\rho_s-1}}\sum_{j=1}^{r_j} q_j^{(i)}(x_1,\ldots,x_{s-1})\mathfrak{p}_{j}^{(i)}(x_1,\ldots,x_s).
\end{equation}
 Note that in \eqref{I-expr}, the coefficients of the polynomials in $\mathfrak{B}$ are polynomials in $S$.\\

\section{Results}\label{res}
Our aim in this section is to construct a basis of  $I$, as an $\ef_q$-vector subspace of $R$ (an $\ef_q$-basis),  from the ideal basis $\mathfrak{B}$ of $I$, in \eqref{Ideal-base}.

\begin{prop}\label{G-set}The set
\begin{equation*}
\mathfrak{B} '= \{x_1^{\alpha_1}\cdots x_{s-1}^{\alpha_{s-1}}\mathfrak{p}\;|\;(\alpha_1,\ldots,\alpha_{s-1})\leqslant_+(\rho_1,\ldots,\rho_{s-1})\;\text{and}\; \mathfrak{p}\in\mathfrak{B}\}
\end{equation*}
is a generating set of $I$, as an $\ef_q$-vector space.
\end{prop}
\begin{proof}It suffices to use \eqref{I-expr} and write\begin{equation*}
q_j^{(i)}(x_1,\ldots,x_{s-1})=\sum_{(\alpha_1,\ldots,\alpha_{s-1})\leqslant_+(\rho_1,\ldots,\rho_{s-1})}q^{(i)}_{j\alpha_1,\ldots,\alpha_{s-1}}x_1^{\alpha_1}\cdots x_{s-1}^{\alpha_{s-1}},
\end{equation*}
where $q^{(i)}_{j\alpha_1,\ldots,\alpha_{s-1}}\in\ef_q$. Then the polynomial  $f$ is written as a linear combination of elements of $\mathfrak{B}'$, with coefficients in $\ef_q$.
\end{proof}
The set $\mathfrak{B}'$  in \ref{G-set} may be too large to be an $\ef_q$- basis of $I$. In other words, the elements of $\mathfrak{B}$ may be linearly dependent.
If this is the case, an $\ef_q$ - basis $B$ of $I$  should be then extracted from $\mathfrak{B}'$.\\

We will find linearly independent elements of $\mathfrak{B}'$ and check whether they form an $\ef_q$ -base of $I$.

According to the notations in \eqref{Ideal-base}, we choose polynomials
\begin{equation*}
  \mathfrak{p}_0(x_1,\ldots,x_s), . . . , \mathfrak{p}_{\rho_s-1}(x_1,\ldots,x_s),
\end{equation*}
where $\mathfrak{p}_k\in\{\mathfrak{p}_1^{(k)},\ldots,\mathfrak{p}^{(k)}_{r_{\rho_s-1}}\}\subset I_{k} $. Let $p_k(x_1,\ldots,x_{s-1})\in S$ be  the coefficient of $\mathfrak{p}_k$ with respect to $x_s^k$ and
  $a_k=\deg p_k$, where the degree is defined as in \eqref{deg-pol}, but, now, in the quotient ring $S$. We have
\begin{equation}\label{pk}
\mathfrak{p}_k(x_1,\ldots,x_s) = \sum_{h=k}^{\rho_{s-1}} p_h^h(x_1,\ldots,x_{s-1})x_s^h,
\end{equation}
with $p_h^h\in S$ and $p_k^k=p_k$.

 \begin{prop}\label{comb-nul-lemme}Let   $l_0(x_1,\ldots,x_{s-1}),\ldots,\l_{\rho_{s-1}}(x_1,\ldots,x_{s-1})$ be polynomials in $\ef_q[x_1,\ldots,x_{s-1}]$ such that $\deg(l_k)<_{+}[(\rho_1,\ldots,\rho_{s-1})-(a_k)]$. Then\begin{align*}
  \sum_{k=0}^{\rho_{s-1}}l_k(x_1,\ldots,x_{s-1})\mathfrak{p}_k(x_1,\ldots,x_{s-1})=0 \Longrightarrow &l_k(x_1,\ldots,x_{s-1})=0\\
&\text{for} \;k=0\ldots,\rho_{s-1}.
\end{align*}
\end{prop}
\begin{proof}Let $l_0(x_1,\ldots,x_{s-1}),\ldots, l_{\rho_{s-1}}(x_1,\ldots,x_{s-1})$ be polynomials  in\\ $\ef_q[x_1,\ldots,x_{s-1}]$ which verify the hypothesis of the proposition, such that\begin{equation*}
 \sum_{k=0}^{\rho_{s-1}}l_k(x_1,\ldots,x_{s-1})\mathfrak{p}_k(x_1,\ldots,x_s)=0.
\end{equation*}
Then
\begin{align*}
&l_0(x_1,\ldots,x_{s-1})p_0(x_1,\ldots,x_{s-1}) =0.
\end{align*}
 Supose that $l_0\neq 0$. By taking the degrees, we have, by Proposition \ref{deg-prod},
\begin{align}\label{deg-prod}
\begin{split}
  \deg(l_0p_0) = &\deg(l_0) + \deg(p_{0})>0.
 \end{split}
\end{align}
But this is impossible for a non-zero polynomial.
It follows that $l_0=0$ and using \eqref{pk}, the same reasoning can  be applied step by step to show that    $l_i$  for   $i=1,\ldots,\rho_{s-1}$.
\end{proof}
\begin{thm}\label{gen-s}With the previous notations, let $B$ be the set
\begin{align*}
B = &\{ x_1^{i^0_1}\ldots x_{s-1}^{i_{s-1}^0}\mathfrak{p_0}(x_1,\ldots,x_s)\;\vert\;(i^0_1,\ldots,i^0_{s-1})<_+(\rho_1,\ldots,\rho_{s-1})-a_0\}\\
\cup &\{ x_1^{i^1_1}\ldots x_{s-1}^{i_{s-1}^1}\mathfrak{p_1}(x_1,\ldots,x_s)\;\vert\;(i^1_1,\ldots,i^1_{s-1})<_+(\rho_1,\ldots,\rho_{s-1})-a_1\}\\
&\ldots\\
\cup &\{ x_1^{i^{r_{n-1}}_1}\ldots x_{s-1}^{i^{r_{n-1}}_{s-1}}\mathfrak{p}_{\rho_s-1}(x_1,\ldots,x_s)\;\vert\;\\
&\quad\quad\quad\quad\quad\quad\quad\quad\quad(i^{r_{n-1}}_1,\ldots,i^{r_{n-1}}_{s-1})<_+(\rho_1,\ldots,\rho_{s-1})-a_{\rho_s-1}\}.
\end{align*}
Then \\
$(1)$ The elements of $B$ are $\ef_q$-linearly  independent.\\
$(2)$ If $|B|={\rm log_q}|I|$, then $B$ is an $\ef_q$-basis of $I$ (where $|I|$ is the cardinality of $I$).
\end{thm}
\begin{proof} (1) We construct the finite sequence of numbers
\begin{align*}
N_k=|\{ x_1^{i^k_1}\ldots x_{s-1}^{i_{s-1}^k}\mathfrak{p}_k(x_1,\ldots,x_s)\;\vert\;(i^k_1,\ldots,i^k_{s-1})<_+(\rho_1,\ldots,\rho_{s-1})-a_k|.
\end{align*}
for $k=0,\ldots,\rho_s-1$.
Now, let $(\alpha_j^k)_{1\leqslant j\leqslant N_k}$ be   sequences of elements of
$\ef_q$ such that
\begin{equation}\label{comb-nulle}
\sum_{k=0}^{\rho_s-1}\sum_{j=1}^{N_k}\alpha_j^k x_1^{i^k_1}\ldots x_{s-1}^{i_{s-1}^k}\mathfrak{p}_k(x_1,\ldots,x_s) = 0
\end{equation}
(this is a  linear combination of elements of $B$ which equals to zero).
By taking
\begin{equation*}
l_k(x_1,\ldots,x_{s-1}) = \sum_{j=1}^{N_k}\alpha_j^k x_1^{i^k_1}\ldots x_{s-1}^{i_{s-1}^k}
\end{equation*}
for $k=0,\ldots,\rho_s-1$, equation  \eqref{comb-nulle} becomes
\begin{equation*}
 \sum_{k=0}^{\rho_s-1}l_k(x_1,\ldots,x_{s-1})\mathfrak{p}_k(x_1,\ldots,x_s)=0.
\end{equation*}
By Proposition \ref{comb-nul-lemme}, we have $l_k(x_1,\ldots,x_{s-1})=0$ for $k=0,\ldots,\rho_{s-1}$, i.e. $\alpha_j^k=0$ for $j=1,\ldots,N_k$. \\
(2)  The ring $R$ is isomorphic to  a subspace of $\ef_q^n$, by the mapping
\begin{align*}
\begin{split}
R&\longleftrightarrow\ef_q^n\\
f(x)=\sum_{\alpha\in\mathcal{G}_s}f_\alpha x^\alpha&\longleftrightarrow (f_\alpha)_{\alpha\in\mathcal{G}_s},
\end{split}
\end{align*}
where   $n=\prod_{i=1}^s{\rho_i}$. Thus, $I$ may be identified with a subspace of $\ef_q^n$, and it is known that  in this case, $\dim I={\rm log_q }|I|$. The elements of the  set $B$  being linearly independent, it follows that $B$ is a basis of $I$, when its when its cardinality equals to ${\rm log_q }|I|$.
\end{proof}
For an $\ef_q$-basis $B=\{g_1(x),\ldots,g_l(x)\}$ of $I$, where, according to \eqref{f-ordered}
\begin{equation*}
   g_\lambda(x)= g_{\lambda\alpha^{(1)}}x^{\alpha^{(1)}}+\cdots+ g_{\lambda\alpha^{(i)}}x^{\alpha^{(1)}}+\cdots+g_{\lambda\alpha^{(\rho_s)}}x^{\alpha^{(n)}}\;\text{for}\;\lambda=1,\ldots,l.
\end{equation*}
A generator matrix for $I$, as a multicyclic code is then
\begin{equation*}
  G= \left(
       \begin{array}{ccccc}
         g_{1\alpha^{(1)}}  & \ldots  &g_{1\alpha^{(\nu)}}& \ldots  & g_{1\alpha^{(n)}}  \\
          \vdots &\vdots   & \vdots  & \vdots&\vdots  \\
          g_{\lambda\alpha^{(1)}}  &  \ldots & g_{\lambda\alpha^{(\nu)}}&\ldots  &  g_{\lambda\alpha^{(n)}}  \\
           \vdots&\vdots   &\vdots   &\vdots &\vdots  \\
           g_{l\alpha^{(1)}}& \ldots   &g_{l\alpha^{(\nu)}}&\ldots  &g_{l\alpha^{(n)}}
       \end{array}
     \right)  \in\ef_q^{l,n},
\end{equation*}
where $\ef_q^{l,n}$ is the set of matrices with $l$ rows and $n$ columns and entries in $\ef_q$. In other words, $G$ is the matrix whose rows are the coefficients of the elements of $B$.
\begin{center}{\textbf{Acknowledgments}}
\end{center}
 The authors would like to thank the referee for careful reading. \\ \\
\vskip 0.4 true cm

\bibliographystyle{amsplain}

\begin{thebibliography}{12}
\bibitem{lar} R. M. Lalasoa, R. Andriamfidisoa and T. J. Rabeherimanana, \emph{Basis of a multicyclic code as an Ideal in $\ef[X_1,\ldots,X_s]/\langle X_1^{\rho_1}-1,\ldots,X_s^{\rho_s}-1 \rangle$},  Journal of Algebra and Related Topics, (2) \textbf{\textbf{6}} (2018), 63--78.
\bibitem{z0} Z. Sepasdar,  \emph{Some Notes on the Characterization of
two dimensional skew cyclic Codes}, Journal of Algebra and Related Topics,
(2) \textbf{4} (2016), 1-8.
\bibitem{z2} Z. Sepasdar, K. Khashyarmanesh, \emph{Characterizations of some two-dimensional cyclic Codes correspond to the Ideals of $\ef[x, y]/\langle x^s-1,y^{2k}-1\rangle $}, Finite Fields and Their Applications {\bf 41} (2016), 97–112.
\bibitem{z1} Z. Sepasdar, \emph{Generator Matrix for two-dimensional cyclic Codes of arbitrary Length}, arXiv:1704.08070v1, [math.AC], 26 Apr 2017.
\end{thebibliography}

\bigskip
\bigskip

{\footnotesize {\bf First Author}\; \\ {Department of
Mathematics}, {University
of Antananarivo, p.O.Box 906,} {101 Antananarivo, Madagascar,}\\
{And \\
Higher Polytechnics Institute of
Madagascar (ISPM)}, { Ambatomaro Antsobolo,}\\ { 101 Antananarivo, Madagascar.}\\
{\tt Email: andriamifidisoa.ramamonjy@univ-antananarivo.mg}\\

{\footnotesize {\bf Second Author}\; \\ {Department of
Mathematics and Computer Science}, {University
of Antananarivo, p.O.Box 906,} {101 Antananarivo, Madagascar.}\\
{\tt Email: larissamarius.lm@gmail.com}\\

{\footnotesize {\bf Third Author}\; \\ {Department of
Mathematics}, {University
of Antananarivo, p.O.Box 906,} {101 Antananarivo, Madagascar.}\\
{\tt Email: rabeherimanana.toussaint@yahoo.fr}\\

\end{document}